\documentclass[graybox,natbib]{SNmult}

\usepackage{type1cm}
\usepackage{makeidx}
\usepackage{graphicx}
\graphicspath{{figures/}}
\usepackage{multicol}
\usepackage[bottom]{footmisc}
\usepackage{newtxtext}
\usepackage[varvw]{newtxmath}

\usepackage{amsmath}
\usepackage{mathtools}
\usepackage{bm}
\usepackage{bbm}
\usepackage{mathrsfs}
\usepackage{booktabs}
\usepackage{longtable}
\usepackage{multirow}
\usepackage{array}
\usepackage{float}
\usepackage{enumitem}
\usepackage[hidelinks]{hyperref}
\usepackage[nameinlink,noabbrev]{cleveref}

\spnewtheorem{assumption}[theorem]{Assumption}{\bfseries}{\rmfamily}

\providecommand{\E}{}
\renewcommand{\E}{\mathbb E}

\makeindex

\begin{document}

\title*{A Neurofinance Framework for Subjective Temporal Perception, Risk, and Investment Behavior}
\titlerunning{Subjective Temporal Perception, Risk, and Investment Behavior}
\author{Pascal Stiefenhofer}
\authorrunning{P. Stiefenhofer}
\institute{Pascal Stiefenhofer \at Newcastle University, Newcastle upon Tyne, UK, \email{pascal.stiefenhofer@newcastle.ac.uk}}

\maketitle

\abstract*{Neurofinance shows that financial valuation depends on evolving neural states,
while temporal experience is itself state dependent. Yet intertemporal models
typically treat time as exogenous and ask how delay affects valuation. This
paper examines the converse question: can valuation-related neural dynamics
generate subjective financial time? We develop a continuous-time model on a
joint financial--neuro-evaluative state space in which subjective financial
time is accumulated valuation along dynamically admissible histories. We
characterise realised temporal-rate dispersion and show that, under matched
financial dynamics, distinct neuro-evaluative trajectories can generate
different subjective financial times. Financial equivalence therefore need not
imply temporal equivalence. Valuation curvature further determines local path
dependence.
A proof-of-concept behavioural--fMRI analysis uses \(1{,}183\) observations,
\(1{,}126\) valuation transitions, and \(798\) financially matched pairs.
Using a medial-prefrontal valuation-state coordinate based on vmPFC and dmPFC
responses, the strongest participant exhibits a positive association between
neural-state separation and subsequent valuation divergence
(\(\rho=0.253,\ p_{\mathrm{perm}}=0.001\)), robust to residual financial
distance (\(\rho=0.251\)) and tighter matching (\(\rho=0.207\)).
Participant effects are heterogeneous, although combined evidence gives
\(p=0.0216\). The data support the neural state-separation premise rather than
directly identifying subjective financial time. The framework establishes a
theoretical and empirically grounded basis for endogenous subjective time in
neurofinance.
}

\abstract{Neurofinance shows that financial valuation depends on evolving neural states,
while temporal experience is itself state dependent. Yet intertemporal models
typically treat time as exogenous and examine how delay affects valuation. This
paper reverses this perspective by asking whether valuation-related neural
dynamics can themselves generate subjective financial time. We develop a continuous-time model on a
joint financial--neuro-evaluative state space in which subjective financial
time is accumulated valuation along dynamically admissible histories. We
characterise realised temporal-rate dispersion and show that, under matched
financial dynamics, distinct neuro-evaluative trajectories can generate
different subjective financial times. Financial equivalence therefore need not
imply temporal equivalence. Valuation curvature further determines local path
dependence.
A proof-of-concept behavioural--fMRI analysis uses \(1{,}183\) observations,
\(1{,}126\) valuation transitions, and \(798\) financially matched pairs.
Using a medial-prefrontal valuation-state coordinate based on vmPFC and dmPFC
responses, the strongest participant exhibits a positive association between
neural-state separation and subsequent valuation divergence
(\(\rho=0.253,\ p_{\mathrm{perm}}=0.001\)), robust to residual financial
distance (\(\rho=0.251\)) and tighter matching (\(\rho=0.207\)).
Participant effects are heterogeneous, although combined evidence gives
\(p=0.0216\). The empirical evidence supports the neural state-separation premise required by the theory, while stopping short of directly identifying subjective financial time. The results therefore provide empirical grounding for a framework in which subjective financial time emerges endogenously from valuation-related neural dynamics.
}


\section{Introduction}
\label{sec:introduction}

Neurofinance shows that financial valuation is dynamically state dependent.
Reward, subjective value, risk, prediction error, loss aversion, framing, and
physiological state recruit interacting neural and biological processes
\citep{CamererLoewensteinPrelec2005,KnutsonAdamsFongHommer2001,
KuhnenKnutson2005,PadoaSchioppaAssad2006,
PreuschoffBossaertsQuartz2006,KnutsonBossaerts2007,
TomFoxTrepelPoldrack2007,PreuschoffQuartzBossaerts2008,
DeMartinoKumaranSeymourDolan2006,DeMartinoCamererAdolphs2010,
MohrBieleHeekeren2010,CoatesHerbert2008,
SapienzaZingalesMaestripieri2009,MiendlarzewskaKometerPreuschoff2019}.
Financial choice therefore reflects both observable economic conditions and
evolving neurocognitive states.

Intertemporal-choice research, however, generally treats time as an input into
valuation. Immediate and delayed rewards recruit partly differentiated
valuation processes \citep{McClureLaibsonLoewensteinCohen2004}, valuation
circuitry tracks discounted subjective value \citep{KableGlimcher2007}, and
subjective timing relates to intertemporal behaviour
\citep{WittmannPaulus2008}. Yet temporal experience itself varies with
attention, affect, interoception, memory, embodiment, and brain--body dynamics
\citep{BuhusiMeck2005,WittmannVanWassenhove2009,Wittmann2013,
TsaoYousefzadehMeckMoserMoser2022,Buzsaki2026}.

This creates an asymmetry. Neurofinance studies how time affects valuation,
but not whether valuation dynamics themselves generate temporal structure.
If both valuation and temporal experience depend on neurocognitive state, then
financially identical histories need not be temporally equivalent. We define
\emph{subjective financial time} as the temporal progression generated by
accumulated valuation along an evolving financial--neuro-evaluative history,
rather than by the passage of calendar time. The paper therefore asks when
distinct neuro-evaluative dynamics can generate different subjective financial
times along the same financial trajectory, and when accumulated subjective
financial time is independent of the path taken. The usual direction of analysis is thus reversed, with valuation dynamics
treated as primitive and subjective financial time emerging endogenously.

The decision maker occupies the joint state space
\(\overline M=\overline M_F\times\overline M_E\), where \(x_F\) denotes the
financial state and \(x_E\) the latent neuro-evaluative state. Admissible
evolution is governed by a differential inclusion together with a valuation
one-form \(\omega\). For an admissible history \(X\), subjective financial time
is defined by
\[
\tau_X(t)
=
\tau_0+
\int_0^t
\omega_{X(s)}\!\left(\dot X(s)\right)\,ds.
\]
Calendar time \(t\) indexes physical evolution, whereas \(\tau_X(t)\) records
the temporal progression induced by valuation dynamics along the realised
history. It is therefore neither a neural clock nor a modification of
chronological time, but an endogenous temporal quantity generated by the
agent's evolving evaluative state.

The theory yields three principal results. Theorem~\ref{thm:dispersion}
characterises the realisable spectrum of subjective-time rates
\(\Lambda(x)\) and its dispersion \(\Delta_\tau(x)\).
Theorem~\ref{thm:neuro-temporal-divergence} establishes that, under financially
matched realisability,
\[
X_F^+(t)=X_F^-(t)
\qquad\not\Rightarrow\qquad
\tau_{X^+}(t)=\tau_{X^-}(t),
\]
so identical financial evolution need not generate identical subjective
financial time. Theorem~\ref{thm:local-structure} identifies valuation
curvature, \(\Omega=d\omega\), as the local source of path dependence in
accumulated subjective financial time. Homogeneous time is recovered as the
special case in which temporal-rate dispersion and valuation curvature both
vanish and the resulting unique temporal rate is normalised to one.

The theoretical analysis extends geometric approaches to rational agency
\citep{Stiefenhofer2026GeometryTime,
Stiefenhofer2026EconomicTemporality,
Stiefenhofer2026KantianTime} by augmenting the financial state space with an
explicit neuro-evaluative state. Admissible financial--neuro-evaluative
histories are represented by a differential inclusion, while a valuation
one-form determines the accumulation of subjective financial time along those
histories. This construction permits temporal-rate dispersion, financial
matching across distinct neuro-evaluative trajectories, and local path
dependence to be characterised mathematically.

The empirical analysis provides a proof-of-concept test of the neural
state-separation premise using the behavioural--fMRI intertemporal-choice
experiment of \citet{PivaEtAl2019}. Event-related vmPFC, dmPFC, and TPJ
responses provide observable neural measurements, with
\[
\nu=z(\mathrm{vmPFC})-z(\mathrm{dmPFC})
\]
used as the primary neuro-evaluative coordinate. Observations are financially
matched on monetary amounts, objective delays, functional run, and Self/Other
condition. Within matched pairs, neural-state separation is related to
subsequent valuation divergence, with permutation inference and alternative
matching restrictions used to assess robustness.

The empirical evidence does not measure \(\tau\), identify \(\dot{\tau}\),
estimate temporal dispersion, or test curvature, and the observed neural
coordinate should not be identified with the theoretical state \(x_E\).
Instead, it establishes the narrower empirical premise required by the theory,
namely that conditioning closely on observable financial state need not
eliminate neuro-evaluative variation associated with subsequent valuation.
A direct empirical test of generated subjective financial time would require
longitudinal measurement of financial states, neuro-evaluative states,
valuation dynamics, and subjective temporal experience.

Section~\ref{sec:model} develops the model,
Section~\ref{sec:analysis} establishes the results,
Section~\ref{sec:empirical} presents the behavioural--fMRI analysis,
Section~\ref{sec:discussion} discusses the implications,
Section~\ref{sec:conclusion} concludes, and
Appendix~\ref{app:appendix} provides the supplementary material.

\section{Model}
\label{sec:model}

\subsection{Financial--Neuro-Evaluative State Geometry}
\label{subsec:state-geometry}

Let \(\overline M_F\) and \(\overline M_E\) be connected finite-dimensional
\(C^2\) manifolds and define
\(\overline M=\overline M_F\times\overline M_E\). The non-empty closed set
\(M\subseteq\overline M\) is the feasible state space, with
\(x=(x_F,x_E)\in M\) denoting financial state \(x_F\) and latent
neuro-evaluative state \(x_E\).

The component \(x_E\) summarises valuation-relevant neurocognitive and
physiological conditions, including reward valuation, salience, arousal,
interoception, perceived risk, affect, memory, and executive control
\citep{KnutsonAdamsFongHommer2001,KuhnenKnutson2005,
PreuschoffBossaertsQuartz2006,TomFoxTrepelPoldrack2007,
PreuschoffQuartzBossaerts2008,Wittmann2013,
MiendlarzewskaKometerPreuschoff2019}. It is a functional rather than
anatomical state. The product structure implies
\(T_x\overline M=T_{x_F}\overline M_F\oplus T_{x_E}\overline M_E\), so
\(v=(v_F,v_E)\). This permits financial motion \(v_F\) to be fixed while
neuro-evaluative motion \(v_E\) varies, without assuming dynamic independence.
Fix a complete Riemannian metric on \(\overline M\). For its induced distance
\(d_{\overline M}\), define the Bouligand contingent tangent cone by
\begin{equation}
T_M(x)
=
\left\{
v\in T_x\overline M:
\liminf_{h\downarrow0}
\frac{d_{\overline M}(\exp_x(hv),M)}{h}=0
\right\}.
\label{eq:contingent-cone}
\end{equation}
For \(x\in\operatorname{int}M\), \(T_M(x)=T_x\overline M\). Possible
functional coordinates include
\(x_F=(w,\pi,\ell,m,z)\) for wealth, portfolio position, liquidity, market
conditions, and other financial variables, and
\(x_E=(\nu,\sigma,\iota,\kappa,\eta,\rho)\) for valuation, salience or arousal,
interoception, executive control, affect, and perceived risk. These coordinates
are illustrative and impose neither a particular coordinate representation nor
a fixed anatomical interpretation. Latent neuro-evaluative state is
distinguished from its empirical measurement through a measurement map
\[
H_E:\overline M_E\longrightarrow\mathcal Y_E,
\]
where \(H_E(x_E)\) denotes a neural, physiological, or behavioural measurement
of \(x_E\). Empirical proxies such as BOLD responses are therefore measurements
of neuro-evaluative state rather than the state itself. Potential state motion is
represented by
\(F:M\rightrightarrows T\overline M\), with
\(\varnothing\neq F(x)\subseteq T_x\overline M\). The feasible motions that
also satisfy the state constraint are collected in the viable feasible
correspondence
\[
F_M(x)
=
F(x)\cap T_M(x).
\]
Thus \(F\) describes potential state motion, while \(F_M\) restricts that
motion to directions that are viable in \(M\). The additional valuation
restriction defining the admissible correspondence \(A\) is introduced in
Subsection~\ref{subsec:valuation-geometry}.

\medskip

\noindent
\textbf{Assumption A1 (Feasible dynamics).}
For every \(x\in M\), \(F(x)\) is non-empty, compact, and convex. The
correspondence \(F\) is upper hemicontinuous in the tangent-bundle topology
and locally bounded.

\medskip

\noindent
\textbf{Assumption A2 (Viability).}
For every \(x\in M\),
\[
F_M(x)
=
F(x)\cap T_M(x)
\neq\varnothing.
\]

\medskip

\noindent
\textbf{Assumption A3 (Growth).}
There exist \(a,b\geq0\) and \(x^\circ\in\overline M\) such that
\[
\sup_{v\in F(x)}\|v\|
\leq
a+b\,d_{\overline M}(x,x^\circ),
\qquad x\in M.
\]

Assumptions A1--A3 provide regularity, viability, and growth control without
imposing separability or uniqueness of financial and neuro-evaluative
dynamics. In particular, the framework permits multiple feasible and viable
directions from a common state, allowing subsequent valuation restrictions to
select among dynamically distinct financial--neuro-evaluative motions.

\subsection{Valuation Geometry}
\label{subsec:valuation-geometry}

Let \(U\supseteq M\) be open and let
\(\omega\in\Omega^1(U)\) be a \(C^1\) valuation one-form. For
\(v=(v_F,v_E)\in T_x\overline M\),
\[
\omega_x(v)
=
\omega_{F,x}(v_F)+\omega_{E,x}(v_E),
\]
where both components may depend on the full state \(x=(x_F,x_E)\).
Consequently, identical financial motion need not generate identical
valuation rates when neuro-evaluative state or motion differs.

\begin{definition}[Valuation geometry]
The pair \((M,\omega)\) is the \emph{valuation geometry}, and
\[
\Omega=d\omega
\]
is its \emph{valuation curvature}.
\end{definition}

If \(\omega=dV\) for some scalar field \(V\), then along every piecewise
\(C^1\) curve \(X\),
\[
\int_0^t
\omega_{X(s)}\!\left(\dot X(s)\right)\,ds
=
V(X(t))-V(X(0)).
\]
More generally, \(d\omega=0\) implies local exactness on contractible
neighbourhoods. Thus \(d\omega\) measures the local obstruction to
endpoint-only valuation accumulation, while global path independence may
additionally depend on the topology of the state space.

The viable feasible correspondence \(F_M\) identifies state motions permitted
by the dynamics and the state constraint. Valuation geometry imposes the
additional orientation condition
\(\omega_x(v)\geq0\). The resulting admissible correspondence is
\[
A(x)
=
F_M(x)\cap
\left\{
v\in T_x\overline M:
\omega_x(v)\geq0
\right\}.
\]
Hence admissible motion is simultaneously feasible, viable, and
non-decreasing in evaluative progression. The inequality permits temporal
stasis, \(\omega_x(v)=0\), while excluding reversal of the subjective
financial-time orientation introduced below.

\medskip

\noindent
\textbf{Assumption A4 (Admissible dynamics).}
The correspondence \(A\) has non-empty compact convex values, is upper
hemicontinuous, and satisfies
\[
\sup_{v\in A(x)}\|v\|
\leq
c\left(1+d_{\overline M}(x,\bar x)\right),
\qquad x\in M,
\]
for some \(c>0\) and \(\bar x\in M\).

Assumption A4 ensures that imposing the valuation orientation preserves the
regularity, non-emptiness, and growth control required for admissible
dynamics. These properties do not follow from Assumptions A1--A3 alone,
because the restriction \(\omega_x(v)\geq0\) may exclude otherwise viable
feasible velocities.

For \(x_0\in M\) and \(T>0\), define the family of admissible histories by
\[
\mathcal S_T(x_0)
=
\left\{
X\in AC([0,T],M):
X(0)=x_0,\ 
\dot X(t)\in A(X(t))
\text{ for a.e. }t\in[0,T]
\right\}.
\]
Writing \(X=(X_F,X_E)\) permits comparison of admissible histories with
\(X_F^+(t)=X_F^-(t)\) and \(X_E^+(t)\neq X_E^-(t)\), thereby isolating the
neuro-evaluative contribution under matched financial motion.

The instantaneous valuation-rate spectrum is
\begin{equation}
\Lambda(x)
=
\omega_x(A(x))
=
[\underline\lambda(x),\overline\lambda(x)]
\subseteq[0,\infty),
\label{eq:rate-spectrum}
\end{equation}
where
\(\underline\lambda(x)=\min_{v\in A(x)}\omega_x(v)\) and
\(\overline\lambda(x)=\max_{v\in A(x)}\omega_x(v)\). Its dispersion is
\begin{equation}
\Delta_\tau(x)
=
\operatorname{diam}\Lambda(x)
=
\overline\lambda(x)-\underline\lambda(x).
\label{eq:dispersion}
\end{equation}

Thus \(\Lambda(x)\) gives the admissible instantaneous valuation rates and
\(\Delta_\tau(x)\) their dispersion. Their dynamic realisability is established
below.

\subsection{Conditional Neuro-Evaluative Temporal Rates}
\label{subsec:conditional-neuro-rates}

Let \(\operatorname{pr}_F:T_x\overline M\to T_{x_F}\overline M_F\) and define
the admissible financial velocities by \(A_F(x)=\operatorname{pr}_F A(x)\).
For \(u_F\in A_F(x)\), the conditional neuro-evaluative velocities are
\[
A_E(x\mid u_F)
=
\{v_E\in T_{x_E}\overline M_E:(u_F,v_E)\in A(x)\}.
\]
This set is non-empty, compact, and convex. The corresponding conditional subjective-time-rate spectrum is
\begin{equation}
\Lambda_E(x\mid u_F)
=
\{\omega_x(u_F,v_E):v_E\in A_E(x\mid u_F)\}
=
[\underline\lambda_E(x\mid u_F),
\overline\lambda_E(x\mid u_F)].
\label{eq:conditional-neuro-spectrum}
\end{equation}
Since
\(\omega_x(u_F,v_E)=\omega_{F,x}(u_F)+\omega_{E,x}(v_E)\), the financial
contribution is fixed across the conditional fibre. Define conditional neuro-evaluative temporal dispersion by
\begin{equation}
\Delta_{\tau,E}(x\mid u_F)
=
\operatorname{diam}\Lambda_E(x\mid u_F)
=
\overline\lambda_E(x\mid u_F)
-
\underline\lambda_E(x\mid u_F).
\label{eq:neuro-dispersion}
\end{equation}
Equivalently,
\[
\Delta_{\tau,E}(x\mid u_F)
=
\max_{v_E\in A_E(x\mid u_F)}\omega_{E,x}(v_E)
-
\min_{v_E\in A_E(x\mid u_F)}\omega_{E,x}(v_E).
\]

Thus \(\Delta_{\tau,E}(x\mid u_F)\) measures temporal-rate dispersion generated
by neuro-evaluative variation after fixing both \(x\) and financial motion
\(u_F\). In particular, \(\Delta_{\tau,E}(x\mid u_F)>0\) implies distinct
admissible neuro-evaluative velocities with different subjective-time rates
under identical instantaneous financial motion. Since \(\Lambda_E(x\mid u_F)\subseteq\Lambda(x)\),
\(\Delta_{\tau,E}(x\mid u_F)\leq\Delta_\tau(x)\). Hence
\(\Delta_\tau(x)\) measures total temporal-rate dispersion, whereas
\(\Delta_{\tau,E}(x\mid u_F)\) isolates its neuro-evaluative component
conditional on financial motion. Dynamic realisation of this dispersion is
established in Theorem~\ref{thm:neuro-temporal-divergence}.

\subsection{Subjective Financial Time}
\label{subsec:subjective-financial-time}

Admissible dynamics satisfy
\(\dot X(t)\in A(X(t))\) a.e., with \(X(0)=x_0\in M\).
The interpretation of evaluative progression as agent-relative time follows
Stiefenhofer and Neesham (2026), who show that any temporal representation
consistent with evaluative admissibility coincides, up to strictly monotone
transformation, with realised evaluative progression. In the present setting,
\(\omega_x(v)\) is the local evaluative rate associated with admissible motion
\(v\in A(x)\). We therefore take its path integral as the canonical
normalisation of subjective financial time.

\begin{definition}[Subjective financial time]
For \(X\in\mathcal S_T(x_0)\), subjective financial time is
\begin{equation}
\tau_X(t)
=
\tau_0+
\int_0^t
\omega_{X(s)}\!\left(\dot X(s)\right)\,ds.
\label{eq:time}
\end{equation}
\end{definition}

Thus \(\tau_X\) is generated by evaluative progression along the realised joint
trajectory rather than imposed as physical clock time. Under the
financial--neuro-evaluative decomposition,
\[
\tau_X(t)-\tau_0
=
\int_0^t\omega_{F,X(s)}(\dot X_F(s))\,ds
+
\int_0^t\omega_{E,X(s)}(\dot X_E(s))\,ds,
\]
so neuro-evaluative motion contributes directly to accumulated subjective
financial time. Since \(X\) is absolutely continuous and admissibility requires
\(\omega_{X(t)}(\dot X(t))\geq0\), \(\tau_X\) is absolutely continuous and
non-decreasing, with
\begin{equation}
\dot\tau_X(t)
=
\omega_{X(t)}(\dot X(t))
\in\Lambda(X(t))
\quad\text{a.e.}
\label{eq:time-rate}
\end{equation}
Hence \(\Delta_\tau(x)>0\) permits distinct instantaneous temporal rates, while
\(\Delta_{\tau,E}(x\mid u_F)>0\) permits such differences under fixed financial
motion. Their dynamic realisability is established below.

\begin{definition}[Temporal equivalence]
Admissible histories \(X_1,X_2\) with common temporal normalisation are
\emph{temporally equivalent}, written \(X_1\sim_\tau X_2\), if
\[
\tau_{X_1}(t)=\tau_{X_2}(t)
\qquad\text{for every }t.
\]
\end{definition}

Temporal equivalence concerns accumulated subjective time, not equality of the
underlying histories. If
\(\omega_{X(t)}(\dot X(t))\geq\varepsilon>0\) a.e., then \(\tau_X\) is
strictly increasing and admits an absolutely continuous inverse \(t=t(\tau)\).
The reparametrised history \(\widehat X(\tau)=X(t(\tau))\) satisfies
\[
\frac{d\widehat X}{d\tau}
=
\frac{\dot X(t(\tau))}
{\omega_{\widehat X(\tau)}(\dot X(t(\tau)))}
\quad\text{a.e.}
\]
Thus subjective financial time may itself parameterise a trajectory whenever
its realised rate is bounded away from zero. The primitives \((M,F,\omega)\)
induce \(A\), admissible histories \(\mathcal S_T(x_0)\), subjective time
\(\tau_X\), rate spectra \(\Lambda\) and \(\Lambda_E\), dispersions
\(\Delta_\tau\) and \(\Delta_{\tau,E}\), and curvature \(\Omega=d\omega\).
The results below characterise their dynamically realised temporal structure.

\section{Analysis and Fundamental Results}
\label{sec:analysis}

We now establish dynamic realisability. Existence of admissible histories
ensures feasible dynamics, while realisability of specified velocities is
required to interpret \(\Lambda(x)\) and \(\Lambda_E(x\mid u_F)\) as realised
subjective-time rates.

\subsection{Existence, Continuation, and Dynamic Realisability}
\label{subsec:existence}

\begin{theorem}[Existence and global continuation of admissible histories]
\label{thm:existence}

Suppose Assumption A4 holds. Then, for every \(x_0\in M\), there exist
\(\delta>0\) and \(X\in AC([0,\delta],M)\) such that
\begin{equation}
X(0)=x_0,
\qquad
\dot X(t)\in A(X(t))
\quad
\text{for a.e. }t\in[0,\delta].
\label{eq:local-admissible-solution}
\end{equation}
Every maximal admissible solution is defined on \([0,\infty)\). Consequently,
for every finite \(T>0\),
\begin{equation}
\mathcal S_T(x_0)\neq\varnothing .
\label{eq:nonempty-solution-family}
\end{equation}
\end{theorem}

\begin{proof}
By Assumption~A4, \(A:M\rightrightarrows T\overline M\) has non-empty compact
convex values, is upper hemicontinuous, locally bounded, and satisfies
\(A(x)\subseteq T_M(x)\). Fix \(x_0\in M\) and a \(C^2\) chart
\(\phi:U_0\to V_0\subseteq\mathbb R^n\) containing \(x_0\). For
\(y=\phi(x)\), set
\[
\widetilde A(y)=D\phi(x)A(x).
\]
Invariance of the contingent cone under \(C^1\) diffeomorphisms gives
\[
D\phi(x)T_M(x)
=
T_{\phi(M\cap U_0)}(\phi(x)),
\]
hence
\(\widetilde A(\phi(x))
\subseteq T_{\phi(M\cap U_0)}(\phi(x))\).
The induced bundle isomorphism preserves upper hemicontinuity, compactness,
convexity, and local boundedness. The finite-dimensional viability theorem
therefore yields a local absolutely continuous solution \(Y\) satisfying
\(\dot Y(t)\in\widetilde A(Y(t))\) a.e. Thus
\(X=\phi^{-1}\circ Y\) satisfies \eqref{eq:local-admissible-solution}. For global continuation, let \(X:[0,T_{\max})\to M\) be maximal and set
\(\rho(t)=d_{\overline M}(X(t),\bar x)\), where \(\bar x\) is the reference
state in Assumption~A4. Since distance is \(1\)-Lipschitz and
\(\|\dot X(t)\|\leq c(1+\rho(t))\) a.e.,
\[
\rho(t)
\leq
\rho(0)+c\int_0^t(1+\rho(s))\,ds.
\]
Gronwall's inequality yields
\begin{equation}
\rho(t)
\leq
(1+\rho(0))e^{ct}-1,
\qquad t<T_{\max}.
\label{eq:distance-bound}
\end{equation}
If \(T_{\max}<\infty\), \eqref{eq:distance-bound} implies
\(\|\dot X(t)\|\leq K_{\max}\) a.e., and therefore
\[
d_{\overline M}(X(t),X(s))
\leq
K_{\max}|t-s|.
\]
Hence \(X(t)\) is Cauchy as \(t\uparrow T_{\max}\). Completeness of
\(\overline M\) gives \(X(t)\to x^\ast\in\overline M\), and closedness of
\(M\) implies \(x^\ast\in M\). Local existence at \(x^\ast\) then extends
\(X\) beyond \(T_{\max}\), contradicting maximality. Thus
\(T_{\max}=\infty\), and restriction to any finite interval proves
\eqref{eq:nonempty-solution-family}.
\end{proof}

\begin{proposition}[Temporal process induced by admissible evolution]
\label{prop:temporal-process}

Every admissible history
\(X=(X_F,X_E)\in\mathcal S_T(x_0)\), together with an initial temporal
normalisation \(\tau_0\in\mathbb R\), determines a unique absolutely continuous
subjective financial-time process
\[
\tau_X(t)
=
\tau_0+
\int_0^t
\omega_{X(s)}\!\left(\dot X(s)\right)\,ds.
\]
Moreover, \(\tau_X\) is non-decreasing and
\begin{equation}
\dot\tau_X(t)
=
\omega_{X(t)}\!\left(\dot X(t)\right)
=
\omega_{F,X(t)}\!\left(\dot X_F(t)\right)
+
\omega_{E,X(t)}\!\left(\dot X_E(t)\right)
\in
\Lambda(X(t))
\quad
\text{for a.e. }t\in[0,T].
\label{eq:temporal-process-rate}
\end{equation}
\end{proposition}

\begin{proof}
Since \(X\) is absolutely continuous, \(\|\dot X\|\in L^1([0,T])\).
Its image is compact, so continuity of \(\omega\) gives \(C_X<\infty\) with
\(\|\omega_{X(t)}\|_{\mathrm{op}}\leq C_X\). Hence
\[
|\omega_{X(t)}(\dot X(t))|
\leq C_X\|\dot X(t)\|
\in L^1([0,T]).
\]
Therefore \(\tau_X\in AC([0,T])\) and
\(\dot\tau_X(t)=\omega_{X(t)}(\dot X(t))\) a.e. Since
\(\dot X(t)\in A(X(t))\), admissibility gives
\(\dot\tau_X(t)\geq0\), so \(\tau_X\) is non-decreasing. Writing \(\dot X=(\dot X_F,\dot X_E)\) gives
\[
\dot\tau_X(t)
=
\omega_{F,X(t)}(\dot X_F(t))
+
\omega_{E,X(t)}(\dot X_E(t))
\in\Lambda(X(t))
\quad\text{a.e.},
\]
which proves \eqref{eq:temporal-process-rate}.
\end{proof}

Proposition~\ref{prop:temporal-process} shows that every admissible history
generates non-decreasing subjective financial time. The remaining question is
whether individual rates in the instantaneous spectra are dynamically
realisable.

\begin{definition}[Local \(C^1\)-velocity realisability]
\label{def:velocity-realisability}
A velocity \(v\in A(x)\) is locally \(C^1\)-realisable at \(x\) if some
\(C^1\) admissible history \(X:[0,\delta]\to M\), \(\delta>0\), satisfies
\(X(0)=x\) and \(\dot X(0)=v\). The correspondence \(A\) is locally
\(C^1\)-velocity-realisable at \(x\) if every \(v\in A(x)\) is so realisable.
\end{definition}

This strengthens Theorem~\ref{thm:existence} by requiring a specified
\(v\in A(x)\), rather than some admissible velocity, to initiate an actual
history. For neuro-evaluative comparisons, realisability must additionally be
conditional on financial motion.

\begin{definition}[Conditional neuro-evaluative realisability]
\label{def:neuro-realisability}
Let \(x\in M\) and \(u_F\in A_F(x)\). A velocity
\(v_E\in A_E(x\mid u_F)\) is locally neuro-evaluatively realisable at
\((x,u_F)\) if some \(C^1\) admissible history
\(X=(X_F,X_E):[0,\delta]\to M\), \(\delta>0\), satisfies
\[
X(0)=x,
\qquad
\dot X_F(0)=u_F,
\qquad
\dot X_E(0)=v_E.
\]
The fibre \(A_E(x\mid u_F)\) is locally neuro-evaluatively realisable if every
\(v_E\in A_E(x\mid u_F)\) is so realisable.
\end{definition}

Conditional realisability fixes financial velocity at the comparison state,
but not the subsequent financial trajectory. Matched-financial realisability,
introduced below, imposes the latter requirement. mThus \(\Delta_{\tau,E}(x\mid u_F)>0\) identifies instantaneous
neuro-evaluative rate dispersion, conditional realisability converts those
rates into local histories, and matched-financial realisability isolates their
temporal consequences under a common realised financial trajectory.

\subsection{Neuro-Evaluative Temporal Heterogeneity}
\label{subsec:neuro-temporal-heterogeneity}

Positive conditional dispersion
\(\Delta_{\tau,E}(x\mid u_F)>0\) means that, at the same complete state \(x\)
and financial velocity \(u_F\), admissible neuro-evaluative velocities carry
different instantaneous subjective-time rates. This remains a geometric
statement until the corresponding velocities are realised by histories sharing
the same financial trajectory.

\begin{definition}[Financially matched neuro-evaluative realisability]
\label{def:matched-neuro-realisability}

Let \(x=(x_F,x_E)\in M\), \(u_F\in A_F(x)\), and
\(v_E^1,v_E^2\in A_E(x\mid u_F)\). The pair
\((v_E^1,v_E^2)\) is \emph{locally financially matched and
neuro-evaluatively realisable} at \((x,u_F)\) if there exist
\(\delta>0\) and \(C^1\) admissible histories
\(X^i=(X_F^i,X_E^i):[0,\delta]\rightarrow M\), \(i=1,2\), such that
\begin{equation}
X^1(0)=X^2(0)=x,
\qquad
X_F^1(t)=X_F^2(t)
\quad
\text{for all }t\in[0,\delta],
\label{eq:matched-financial-path}
\end{equation}
and
\begin{equation}
\dot X_F^1(0)=\dot X_F^2(0)=u_F,
\qquad
\dot X_E^i(0)=v_E^i,
\qquad i=1,2.
\label{eq:matched-neuro-velocities}
\end{equation}
\end{definition}

This condition is stronger than
Definition~\ref{def:neuro-realisability}: it requires both prescribed initial
neuro-evaluative velocities and a common realised financial path over a
non-degenerate interval. It is sufficient, but not asserted to be necessary,
for the comparative result below.

\begin{theorem}[Neuro-evaluative temporal divergence under matched financial dynamics]
\label{thm:neuro-temporal-divergence}

Let \(x=(x_F,x_E)\in M\) and \(u_F\in A_F(x)\). Suppose
\begin{equation}
\Delta_{\tau,E}(x\mid u_F)>0.
\label{eq:positive-neuro-dispersion}
\end{equation}
Then there exist
\(v_E^-,v_E^+\in A_E(x\mid u_F)\) satisfying
\begin{align}
\omega_x(u_F,v_E^-)
&=
\underline\lambda_E(x\mid u_F),
\label{eq:conditional-minimizer}
\\
\omega_x(u_F,v_E^+)
&=
\overline\lambda_E(x\mid u_F).
\label{eq:conditional-maximizer}
\end{align}
If \((v_E^-,v_E^+)\) is locally financially matched and neuro-evaluatively
realisable at \((x,u_F)\), let
\(X^-=(X_F,X_E^-)\) and \(X^+=(X_F,X_E^+)\) be corresponding histories with
\(\tau_{X^-}(0)=\tau_{X^+}(0)=\tau_0\). Then
\begin{equation}
\dot\tau_{X^+}(0)-\dot\tau_{X^-}(0)
=
\omega_{E,x}(v_E^+)-\omega_{E,x}(v_E^-)
=
\Delta_{\tau,E}(x\mid u_F)
>0.
\label{eq:initial-neuro-rate-gap}
\end{equation}
Moreover, there exists \(\delta_0\in(0,\delta]\) such that
\begin{equation}
\tau_{X^+}(t)-\tau_{X^-}(t)
\geq
\frac{1}{2}
\Delta_{\tau,E}(x\mid u_F)t
>0
\qquad
\text{for every }t\in(0,\delta_0].
\label{eq:persistent-neuro-divergence}
\end{equation}

Thus dynamically realisable differences in neuro-evaluative motion can
generate different subjective financial times from the same complete initial
state while the realised financial trajectory remains identical.
\end{theorem}

\begin{proof}

Because \(u_F\in A_F(x)\), the fibre \(A_E(x\mid u_F)\) is non-empty.
Compactness and convexity of \(A(x)\) imply compactness and convexity of this
fibre. The map \(v_E\mapsto\omega_x(u_F,v_E)\) is continuous and affine, so it
attains its minimum and maximum at some \(v_E^-\) and \(v_E^+\). Hence
\begin{align}
\Delta_{\tau,E}(x\mid u_F)
&=
\omega_x(u_F,v_E^+)-\omega_x(u_F,v_E^-)
\nonumber\\
&=
\omega_{E,x}(v_E^+)-\omega_{E,x}(v_E^-),
\label{eq:neuro-dynamic-gap}
\end{align}
because the common financial term \(\omega_{F,x}(u_F)\) cancels. By matched neuro-evaluative realisability, there exist \(C^1\) admissible
histories \(X^-\) and \(X^+\) satisfying
\(X_F^-(t)=X_F^+(t)\) on \([0,\delta]\) and
\(\dot X^\pm(0)=(u_F,v_E^\pm)\). Define
\[
r_\pm(t)
:=
\omega_{X^\pm(t)}\!\left(\dot X^\pm(t)\right).
\]
Since \(X^\pm\) are \(C^1\) and \(\omega\) is continuous, \(r_\pm\) are
continuous. Proposition~\ref{prop:temporal-process} and
\eqref{eq:neuro-dynamic-gap} give
\[
r_+(0)-r_-(0)
=
\dot\tau_{X^+}(0)-\dot\tau_{X^-}(0)
=
\Delta_{\tau,E}(x\mid u_F)>0.
\]
Continuity therefore yields \(\delta_0\in(0,\delta]\) such that
\[
r_+(t)-r_-(t)
\geq
\frac{1}{2}\Delta_{\tau,E}(x\mid u_F)
\qquad
\text{for all }t\in[0,\delta_0].
\]
Using the common temporal normalisation,
\[
\begin{aligned}
\tau_{X^+}(t)-\tau_{X^-}(t)
&=
\int_0^t
\bigl[r_+(s)-r_-(s)\bigr]\,ds
\\
&\geq
\frac{1}{2}\Delta_{\tau,E}(x\mid u_F)t
>0,
\end{aligned}
\]
which proves \eqref{eq:persistent-neuro-divergence}.
\end{proof}

Theorem~\ref{thm:neuro-temporal-divergence} isolates the source of the initial
temporal separation exactly. Because the histories start from the same
complete state and share the same financial velocity, the term
\(\omega_{F,x}(u_F)\) cancels; the initial rate gap is entirely
neuro-evaluative. Once \(X_E^+(t)\) and \(X_E^-(t)\) separate, the full states
also differ, so subsequent divergence may operate both directly through
\(\omega_E\) and indirectly through the dependence of \(\omega_F\) on the
neuro-evaluative state. The financial path nevertheless remains common, so the
difference originates in neuro-evaluative dynamics.

\begin{corollary}[Local non-reducibility to financial histories]
\label{cor:neuro-nonreducibility}

Under the hypotheses of
Theorem~\ref{thm:neuro-temporal-divergence}, subjective financial time is not
locally determined by the realised financial history alone. In particular,
there is no functional \(\mathfrak T_F\) of the common financial-history
segment such that
\begin{equation}
\tau_X(t)-\tau_0
=
\mathfrak T_F\!\left(X_F|_{[0,t]}\right)
\label{eq:financial-only-representation}
\end{equation}
holds for both \(X=X^+\) and \(X=X^-\) for every sufficiently small \(t>0\).
\end{corollary}

\begin{proof}

For every \(t\in[0,\delta_0]\),
\(X_F^+|_{[0,t]}=X_F^-|_{[0,t]}\). If
\eqref{eq:financial-only-representation} held for both histories, then
\(\tau_{X^+}(t)=\tau_{X^-}(t)\), contradicting
\eqref{eq:persistent-neuro-divergence}.
\end{proof}

The corollary gives the model's non-reducibility result: conditional on the
theorem's hypotheses, financial observables alone are insufficient to recover
subjective financial time. Theorem~\ref{thm:neuro-temporal-divergence}
therefore identifies a \emph{neuro-dynamic effect}: the complete initial state
and realised financial path are fixed, while different neuro-evaluative
velocities generate different temporal histories. This is distinct from a
\emph{neuro-state effect}, in which financial conditions are fixed but the
current neuro-evaluative states themselves differ; that mechanism is considered
next.

\subsection{Dynamic Representation of Temporal Dispersion}
\label{subsec:dispersion-representation}

The spectrum
\(\Lambda(x)=\omega_x(A(x))\) and its diameter
\(\Delta_\tau(x)=\overline\lambda(x)-\underline\lambda(x)\)
describe instantaneous admissible temporal-rate variation. Under local
\(C^1\)-velocity realisability, these geometric objects admit an exact
representation through realised admissible histories.

For \(x\in M\), let
\[
\mathcal R^1(x)
:=
\left\{
X:
\begin{array}{l}
X\in C^1([0,\delta],\overline M)
\text{ for some }\delta>0,\\
X(0)=x,\quad X([0,\delta])\subseteq M,\\
\dot X(t)\in A(X(t))
\text{ for every }t\in[0,\delta]
\end{array}
\right\}.
\]
For every \(X\in\mathcal R^1(x)\),
Proposition~\ref{prop:temporal-process} strengthens pointwise to
\[
\dot\tau_X(t)
=
\omega_{X(t)}\!\left(\dot X(t)\right)
\qquad
\text{for every }t
\]
on the interval of definition.

\begin{theorem}[Dynamic representation of the admissible rate spectrum]
\label{thm:dispersion}

Let \(x_0\in M\), and suppose that \(A\) is locally
\(C^1\)-velocity-realisable at \(x_0\). Then
\begin{equation}
\Lambda(x_0)
=
\left\{
\dot\tau_X(0):
X\in\mathcal R^1(x_0)
\right\}.
\label{eq:representation-spectrum}
\end{equation}
Consequently,
\begin{equation}
\Delta_\tau(x_0)
=
\operatorname{diam}
\left\{
\dot\tau_X(0):
X\in\mathcal R^1(x_0)
\right\}.
\label{eq:representation-dispersion}
\end{equation}

Moreover, there exist
\(X_-,X_+\in\mathcal R^1(x_0)\) such that
\[
\dot\tau_{X_-}(0)
=
\underline\lambda(x_0),
\qquad
\dot\tau_{X_+}(0)
=
\overline\lambda(x_0),
\]
and, under the common temporal normalisation
\(\tau_{X_-}(0)=\tau_{X_+}(0)=\tau_0\),
\begin{equation}
\tau_{X_+}(t)-\tau_{X_-}(t)
=
\Delta_\tau(x_0)t+o(t)
\qquad
(t\downarrow0).
\label{eq:dispersion-expansion}
\end{equation}
Hence
\[
\Delta_\tau(x_0)
=
\lim_{t\downarrow0}
\frac{
\tau_{X_+}(t)-\tau_{X_-}(t)
}{t}.
\]
\end{theorem}

\begin{proof}

Let \(X\in\mathcal R^1(x_0)\). Since
\(\dot X(0)\in A(x_0)\), Proposition~\ref{prop:temporal-process} gives
\[
\dot\tau_X(0)
=
\omega_{x_0}\!\left(\dot X(0)\right)
\in\Lambda(x_0).
\]
Hence
\[
\left\{
\dot\tau_X(0):
X\in\mathcal R^1(x_0)
\right\}
\subseteq
\Lambda(x_0).
\]
Conversely, let \(\lambda\in\Lambda(x_0)\). By
\(\Lambda(x_0)=\omega_{x_0}(A(x_0))\), there exists
\(v\in A(x_0)\) such that
\[
\lambda=\omega_{x_0}(v).
\]
Local \(C^1\)-velocity realisability at \(x_0\) yields
\(X\in\mathcal R^1(x_0)\) with \(\dot X(0)=v\). Therefore,
Proposition~\ref{prop:temporal-process} gives
\[
\dot\tau_X(0)
=
\omega_{x_0}(v)
=
\lambda.
\]
Thus
\[
\Lambda(x_0)
\subseteq
\left\{
\dot\tau_X(0):
X\in\mathcal R^1(x_0)
\right\},
\]
which proves \eqref{eq:representation-spectrum}. Taking diameters gives
\eqref{eq:representation-dispersion}. Because \(A(x_0)\) is compact and
\(v\mapsto\omega_{x_0}(v)\) is continuous, there exist
\(v_-,v_+\in A(x_0)\) satisfying
\[
\omega_{x_0}(v_-)
=
\underline\lambda(x_0),
\qquad
\omega_{x_0}(v_+)
=
\overline\lambda(x_0).
\]
Local \(C^1\)-velocity realisability provides
\(X_-,X_+\in\mathcal R^1(x_0)\) with
\(\dot X_-(0)=v_-\) and \(\dot X_+(0)=v_+\). Hence
\[
\dot\tau_{X_-}(0)
=
\underline\lambda(x_0),
\qquad
\dot\tau_{X_+}(0)
=
\overline\lambda(x_0).
\]
Define
\[
r_\pm(t)
=
\omega_{X_\pm(t)}
\!\left(\dot X_\pm(t)\right).
\]
Since \(X_\pm\) are \(C^1\) and \(\omega\) is continuous,
\(r_\pm\) are continuous at \(t=0\). Therefore
\[
r_\pm(t)=r_\pm(0)+o(1)
\qquad
(t\downarrow0),
\]
and hence
\[
r_+(t)-r_-(t)
=
\Delta_\tau(x_0)+o(1).
\]
Using the common temporal normalisation,
\[
\begin{aligned}
\tau_{X_+}(t)-\tau_{X_-}(t)
&=
\int_0^t
\left[r_+(s)-r_-(s)\right]\,ds \\
&=
\Delta_\tau(x_0)t+o(t),
\end{aligned}
\]
which proves \eqref{eq:dispersion-expansion}. Dividing by \(t>0\) and taking
\(t\downarrow0\) gives
\[
\Delta_\tau(x_0)
=
\lim_{t\downarrow0}
\frac{\tau_{X_+}(t)-\tau_{X_-}(t)}{t},
\]
completing the proof.
\end{proof}

\ 
\begin{corollary}[Realised instantaneous temporal multiplicity]
\label{cor:temporal-multiplicity}

Suppose \(A\) is locally \(C^1\)-velocity-realisable at \(x_0\). Then the
following are equivalent:
\begin{enumerate}
\item \(\Delta_\tau(x_0)>0\);
\item there exist \(X_1,X_2\in\mathcal R^1(x_0)\) such that
\(\dot\tau_{X_1}(0)\neq\dot\tau_{X_2}(0)\);
\item there exist \(X_-,X_+\in\mathcal R^1(x_0)\) such that
\[
\lim_{t\downarrow0}
\frac{\tau_{X_+}(t)-\tau_{X_-}(t)}{t}>0.
\]
\end{enumerate}
\end{corollary}

\begin{proof}

The equivalence of (1) and (2) follows from
\eqref{eq:representation-dispersion}. If
\(\Delta_\tau(x_0)>0\), Theorem~\ref{thm:dispersion} provides
\(X_-,X_+\in\mathcal R^1(x_0)\) such that
\[
\tau_{X_+}(t)-\tau_{X_-}(t)
=
\Delta_\tau(x_0)t+o(t).
\]
Hence
\[
\lim_{t\downarrow0}
\frac{\tau_{X_+}(t)-\tau_{X_-}(t)}{t}
=
\Delta_\tau(x_0)>0,
\]
which proves (3). Conversely, (3) implies
\[
\dot\tau_{X_+}(0)-\dot\tau_{X_-}(0)>0,
\]
so the initial temporal rates are distinct. By
\eqref{eq:representation-dispersion},
\(\Delta_\tau(x_0)>0\).
\end{proof}

\begin{remark}[Interpretation and scope of temporal dispersion]
\label{rem:temporal-dispersion-scope} The quantity \(\Delta_\tau(x)\) measures total instantaneous temporal-rate
dispersion across admissible velocities, whereas
\(\Delta_{\tau,E}(x\mid u_F)\) holds financial velocity fixed and isolates
dispersion generated by admissible neuro-evaluative motion. Since
\(\Lambda_E(x\mid u_F)\subseteq\Lambda(x)\),
\(\Delta_{\tau,E}(x\mid u_F)\leq\Delta_\tau(x)\), so
\(\Delta_\tau(x)>0\) alone does not identify a neuro-evaluative source.
Both quantities are local and first-order. They characterise initial
temporal-rate separation, while finite-horizon divergence depends on the
subsequent evolution of the state, admissible correspondence, and valuation
one-form. Path dependence is distinct and is governed locally by \(d\omega\),
as established below.
\end{remark}

\subsection{Local Structure of Subjective Financial Time}
\label{subsec:local-structure}

The preceding results identify two distinct sources of temporal structure.
The spectrum \(\Lambda(x)\) describes instantaneous subjective-time rates,
with \(\Delta_\tau(x)\) measuring their total dispersion and
\(\Delta_{\tau,E}(x\mid u_F)\) the dispersion remaining after financial motion
is fixed. By contrast, \(d\omega\) governs the local dependence of accumulated
subjective time on the path through financial--neuro-evaluative state space.
The following theorem collects these results and establishes the corresponding
local curvature representation.

For a piecewise \(C^1\) curve \(\gamma:[a,b]\rightarrow U\), define its
valuation-induced temporal accumulation by
\begin{equation}
\mathcal T_\omega(\gamma)
:=
\int_\gamma\omega.
\label{eq:temporal-line-integral}
\end{equation}

\begin{theorem}[Local structure of subjective financial time]
\label{thm:local-structure}

Let \(x\in M\), and suppose that \(A\) is locally
\(C^1\)-velocity-realisable at \(x\). Then

\begin{enumerate}

\item[(i)] The total instantaneous temporal dispersion satisfies
\[
\Delta_\tau(x)
=
\operatorname{diam}
\left\{
\dot\tau_X(0):
X\in\mathcal R^1(x)
\right\}.
\]

\item[(ii)] For every \(u_F\in A_F(x)\),
\[
\Lambda_E(x\mid u_F)\subseteq\Lambda(x),
\qquad
\Delta_{\tau,E}(x\mid u_F)\leq\Delta_\tau(x).
\]

\item[(iii)] Let \(V\subseteq U\) be a sufficiently small contractible
neighbourhood, and let \(\gamma_1,\gamma_2\subset V\) be piecewise \(C^1\)
curves with common initial and terminal points such that the closed curve
\(\gamma_1-\gamma_2\) bounds an oriented piecewise smooth surface
\(S\subset V\). Then
\begin{equation}
\mathcal T_\omega(\gamma_1)
-
\mathcal T_\omega(\gamma_2)
=
\int_S d\omega.
\label{eq:local-structure-curvature}
\end{equation}
Thus \(d\omega\) measures the local failure of endpoint-only temporal
accumulation.

\item[(iv)] If \(d\omega=0\) on a contractible neighbourhood \(V\), then
there exists \(\Theta\in C^2(V)\), unique up to an additive constant, such
that
\[
\omega|_V=d\Theta.
\]
Consequently, every piecewise \(C^1\) curve
\(\gamma:[a,b]\rightarrow V\) satisfies
\begin{equation}
\mathcal T_\omega(\gamma)
=
\Theta(\gamma(b))-\Theta(\gamma(a)).
\label{eq:local-state-time}
\end{equation}

\end{enumerate}

Hence instantaneous temporal-rate heterogeneity and historical path dependence
are mathematically distinct. The former is represented by \(\Lambda\), while
the latter is governed locally by \(d\omega\).
\end{theorem}

\begin{proof}

Statement (i) follows directly from
Theorem~\ref{thm:dispersion} and
\eqref{eq:representation-dispersion}.

For (ii), fix \(u_F\in A_F(x)\). By definition,
\[
\Lambda_E(x\mid u_F)
=
\left\{
\omega_x(u_F,v_E):
v_E\in A_E(x\mid u_F)
\right\}.
\]
Since \((u_F,v_E)\in A(x)\) for every
\(v_E\in A_E(x\mid u_F)\),
\[
\Lambda_E(x\mid u_F)\subseteq\Lambda(x).
\]
Taking diameters gives
\[
\Delta_{\tau,E}(x\mid u_F)
\leq
\Delta_\tau(x).
\]

For (iii), orient \(S\) so that
\(\partial S=\gamma_1-\gamma_2\). By
\eqref{eq:temporal-line-integral} and Stokes' theorem,
\[
\begin{aligned}
\mathcal T_\omega(\gamma_1)
-
\mathcal T_\omega(\gamma_2)
&=
\int_{\gamma_1}\omega-\int_{\gamma_2}\omega \\
&=
\int_{\partial S}\omega \\
&=
\int_S d\omega.
\end{aligned}
\]
Thus \(d\omega\) measures the local geometric path dependence of the temporal
line integral. When \(\gamma_1\) and \(\gamma_2\) are admissible histories,
the same identity gives the corresponding dynamically realised difference in
subjective financial time, proving
\eqref{eq:local-structure-curvature}.

For (iv), suppose \(d\omega=0\) on the contractible neighbourhood \(V\).
By the Poincar\'e lemma there exists
\(\Theta\in C^2(V)\) such that
\[
d\Theta=\omega|_V.
\]
Hence, for every piecewise \(C^1\) curve
\(\gamma:[a,b]\to V\),
\[
\mathcal T_\omega(\gamma)
=
\int_\gamma\omega
=
\int_\gamma d\Theta
=
\Theta(\gamma(b))-\Theta(\gamma(a)),
\]
which proves \eqref{eq:local-state-time}. If
\(\Theta_1\) and \(\Theta_2\) are two such potentials, then
\(d(\Theta_1-\Theta_2)=0\). Since \(V\) is connected,
\(\Theta_1-\Theta_2\) is constant, so the local temporal potential is unique
up to an additive constant.
\end{proof}

\begin{corollary}[Locally unique path-independent subjective time]
\label{cor:unique-path-independent-time}

Let \(V\subseteq U\) be contractible and suppose
\[
\Delta_\tau(x)=0
\qquad
\text{for every }x\in V\cap M,
\]
and
\[
d\omega=0
\qquad
\text{on }V.
\]
Then there exists \(\Theta\in C^2(V)\) such that every admissible history
\(X\) contained in \(V\cap M\) satisfies
\begin{equation}
\tau_X(t)-\tau_X(0)
=
\Theta(X(t))-\Theta(X(0)).
\label{eq:unique-path-independent-time}
\end{equation}
Moreover, at every \(x\in V\cap M\), all admissible velocities generate the
same instantaneous subjective-time rate.

Thus subjective financial time is locally single-rate and path independent.
Neither property implies equality with calendar time.
\end{corollary}

\begin{proof}

Since \(\Delta_\tau(x)=0\), each non-empty compact interval
\(\Lambda(x)\) is a singleton,
\[
\Lambda(x)=\{\lambda(x)\},
\qquad x\in V\cap M.
\]
Hence all admissible velocities at a given state generate the same
instantaneous temporal rate.

Since \(d\omega=0\) on \(V\),
Theorem~\ref{thm:local-structure}(iv) gives a potential
\(\Theta\in C^2(V)\) satisfying \(d\Theta=\omega|_V\). Applying
\eqref{eq:local-state-time} to an admissible history segment gives
\eqref{eq:unique-path-independent-time}.
\end{proof}

\begin{corollary}[Recovery of normalised calendar time]
\label{cor:calendar-time}

Suppose, in addition to the hypotheses of
Corollary~\ref{cor:unique-path-independent-time}, that
\[
\Lambda(x)=\{1\}
\qquad
\text{for every }x\in V\cap M.
\]
Then every admissible history \(X\) contained in \(V\cap M\) satisfies
\begin{equation}
\tau_X(t)
=
\tau_X(0)+t.
\label{eq:calendar-recovery}
\end{equation}
\end{corollary}

\begin{remark}[Scope of the local structure]
The pair \((\Lambda,d\omega)\) does not completely determine the admissible
financial--neuro-evaluative dynamics. Distinct correspondences \(A\) may
generate identical rate spectra and valuation curvature while supporting
different state trajectories. Theorem~\ref{thm:local-structure} instead
distinguishes three temporal objects. The spectrum \(\Lambda(x)\) describes
instantaneous temporal-rate heterogeneity,
\(\Delta_{\tau,E}(x\mid u_F)\) isolates neuro-evaluative temporal dispersion
when financial motion is fixed, and \(d\omega\) measures the local obstruction
to path-independent temporal accumulation. These objects characterise distinct
features of subjective financial time and do not constitute a complete
invariant of the underlying dynamics.
\end{remark}

\section{Empirical Analysis}
\label{sec:empirical}

Theorem~\ref{thm:neuro-temporal-divergence} implies that financial state
\(x_F\) need not exhaust valuation dynamics when neuro-evaluative state \(x_E\)
differs. The empirical analysis examines the corresponding state-separation
premise: conditional on closely matched observable financial states, is greater
neural-state separation associated with greater subsequent subjective-valuation
divergence? This is a proof-of-concept test of a premise of the theoretical
mechanism, not a direct test of subjective financial time. In particular, the
data do not identify \(\tau\), \(\dot{\tau}\), \(\omega\), or
\(\Delta_{\tau,E}\).

\subsection{Design}
\label{subsec:empirical-design}

The analysis contains \(1{,}183\) behavioural--fMRI observations from \(5\)
participants and \(40\) runs, yielding \(1{,}126\) valid within-run valuation
transitions and \(798\) financially matched pairs. Data construction,
preprocessing, ROI specification, and quality control are reported in
Appendix~\ref{app:appendix}.

The neural specification follows \citet{PivaEtAl2019}, who report that relative
subjective value is positively associated with vmPFC and negatively associated
with dmPFC activity. Financial state is represented by
\[
X_{F,t}=(M_{L,t},M_{R,t},D_{L,t},D_{R,t},A_t),
\]
where \(M\), \(D\), and \(A_t\) denote monetary amounts, objective delays, and
the Self/Other condition. The empirical neural coordinate is defined as
\[
\nu_t
=
z(\mathrm{vmPFC}_t)-z(\mathrm{dmPFC}_t).
\]
This contrast provides a scalar measure of relative medial-prefrontal activity
motivated by the opposing valuation-related associations reported by
\citet{PivaEtAl2019}. It is an empirical neural measurement and is not
identified with the theoretical neuro-evaluative state \(x_E\), which may in
general be multidimensional. Figure~\ref{fig:roi-data-overlay} reports
event-related BOLD percent signal change (PSC) in the pre-specified vmPFC,
dmPFC, and TPJ regions.

\begin{figure}[t]
\centering
\includegraphics[width=\linewidth]
{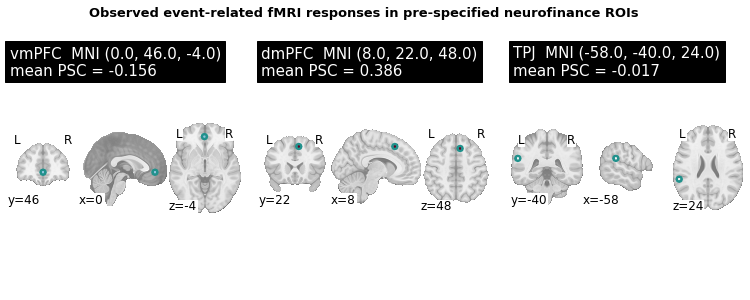}
\caption{\textbf{Event-related fMRI responses in pre-specified ROIs.}
Observed PSC for the \(1{,}183\) retained observations in 5-mm vmPFC
\((0,46,-4)\), dmPFC \((8,22,48)\), and TPJ \((-58,-40,24)\) MNI152
ROIs. Colour represents ROI-level PSC, not voxelwise whole-brain inference.}
\label{fig:roi-data-overlay}
\end{figure}

For consecutive trials, define
\[
\Delta SV_t=SV_{t+1}-SV_t.
\]
Financial states are matched on standardised
\((M_L,M_R,D_L,D_R)\) within participant, run, and agency condition, producing
\[
N_s=(154,162,166,163,153)
\]
matched pairs across participants.

For matched observations \(i,j\), neural-state separation and subsequent
valuation divergence are defined by
\[
D_E(i,j)=|\nu_i-\nu_j|,
\qquad
D_{\Delta SV}(i,j)=|\Delta SV_i-\Delta SV_j|.
\]
The participant-level statistic is
\[
\rho_s
=
\operatorname{corr}_S
\left(D_E,D_{\Delta SV}\right).
\]
Thus \(\rho_s>0\) indicates that, among observations with closely matched
financial states, greater separation in the empirical neural coordinate is
associated with greater subsequent subjective-valuation divergence.

Inference uses \(5{,}000\) structure-preserving permutations within
participant--run--agency cells. Robustness analyses control for residual
financial-state distance, retain only the closest \(50\%\) of financial
matches, and replace medial-prefrontal separation with TPJ separation.
TPJ serves as a decision-relevant comparator because \citet{PivaEtAl2019}
implicate it in social decision processing and valuation-related activity.
Full computational specifications are reported in
Appendix~\ref{app:appendix}.

\subsection{Results}
\label{subsec:empirical-results}

Table~\ref{tab:empirical-main} reports the participant-level results.

\begin{table}[t]
\centering
\caption{Neural-state separation and subsequent valuation divergence}
\label{tab:empirical-main}
\small
\begin{tabular}{crrrrr}
\toprule
Participant & Pairs & \(\rho\) & \(\rho_{\rm partial}\) &
\(p_{\rm perm}\) & \(\rho_{\rm TPJ}\) \\
\midrule
01 & 154 &  0.006 &  0.012 & 0.448 &  0.009 \\
02 & 162 & -0.087 & -0.089 & 0.624 &  0.092 \\
03 & 166 & \textbf{0.253} & \textbf{0.251} &
\textbf{0.001} & 0.205 \\
04 & 163 & -0.053 & -0.053 & 0.734 & -0.010 \\
05 & 153 &  0.097 &  0.106 & 0.139 & -0.096 \\
\midrule
Total & 798 & & & & \\
\bottomrule
\end{tabular}
\end{table}

The participant-level coefficients are
\[
\rho_s=(0.006,-0.087,0.253,-0.053,0.097),
\]
with permutation probabilities
\[
p_s=(0.448,0.624,0.001,0.734,0.139).
\]
Three of the five coefficients are positive, but the estimates exhibit
substantial participant heterogeneity. Fisher combination gives
\(p_{\rm Fisher}=0.0216\). Given the small number of participants and the
heterogeneous coefficient signs, this combined probability is interpreted as
descriptive evidence against the joint null rather than as evidence of a
homogeneous population-level effect.

Participant~03 provides the clearest within-participant evidence, with
\(N=166\), \(\rho=0.253\), \(\rho_{\rm partial}=0.251\), and
\(p_{\rm perm}=0.001\). Adjustment for residual financial-state distance changes
the coefficient by only \(|\Delta\rho|=0.0026\), approximately \(1.0\%\).
Restricting identification to the closest \(50\%\) of financial matches gives
\(N_{\rm strict}=83\), \(\rho_{\rm strict}=0.207\), and
\(p_{\rm desc}=0.061\). Hence
\[
\frac{\rho_{\rm strict}}{\rho}
=
\frac{0.207}{0.253}
=
0.818,
\]
so \(81.8\%\) of the primary coefficient remains when identification is
restricted to the closest half of the matched observations.

For Participant~03, the TPJ comparator gives
\(\rho_{\rm TPJ}=0.205\), compared with
\(\rho_{\rm MPFC}=0.253\), a difference of \(0.048\). Under strict matching,
the corresponding coefficients are
\(\rho_{\rm TPJ,strict}=0.143\) and
\(\rho_{\rm MPFC,strict}=0.207\), increasing the difference to \(0.064\).
The medial-prefrontal coordinate therefore produces the larger association
under both specifications. The positive TPJ coefficient, however, precludes
an interpretation based on anatomical selectivity.

Figure~\ref{fig:sub03-brain-overlay} shows Participant~03's event-related fMRI
responses in the pre-specified vmPFC, dmPFC, and TPJ regions of interest,
displayed in three orthogonal anatomical views. Mean percentage signal change
is \(-0.310\), \(0.188\), and \(-0.274\), respectively. The anatomical
overlays represent ROI-level responses and should not be interpreted as
voxelwise whole-brain statistical maps.

\begin{figure}[t]
\centering
\includegraphics[width=\linewidth]
{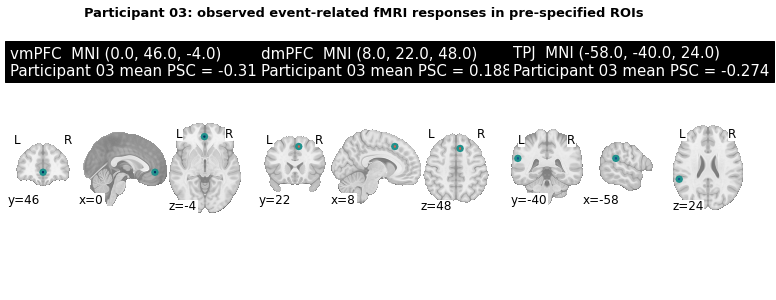}
\caption{\textbf{Participant~03 ROI responses.}
Observed event-related PSC in the pre-specified vmPFC, dmPFC, and TPJ ROIs
(\(N=166\), \(\rho=0.253\), \(p_{\rm perm}=0.001\)).
Colour denotes ROI-level PSC, not voxelwise whole-brain inference.}
\label{fig:sub03-brain-overlay}
\end{figure}

The empirical data sequence is
\[
1{,}200\rightarrow1{,}183\rightarrow1{,}126\rightarrow798,
\]
corresponding to candidate trials, quality-controlled observations, valid
valuation transitions, and financially matched pairs. For Participant~03, the
coefficient sequence
\[
0.253\rightarrow0.251\rightarrow0.207
\]
under the primary, residual-distance-adjusted, and strict-match specifications
shows limited attenuation as financial matching becomes more restrictive.

The results provide proof-of-concept evidence for the state-separation premise
underlying \(x=(x_F,x_E)\). Similarity in observable financial state \(x_F\)
need not imply similarity in neural measurements associated with
neuro-evaluative state. For Participant~03, greater neural-state separation
under closely matched financial conditions is associated with greater
subsequent subjective-valuation divergence, and this association remains
largely preserved under tighter matching.

The empirical claim is narrower than the theoretical result. The neural
measurements do not identify the latent state \(x_E\), and the analysis does
not estimate \(\tau\), \(\dot{\tau}\), \(\omega\), or
\(\Delta_{\tau,E}\). It therefore does not directly test
Theorem~\ref{thm:neuro-temporal-divergence}. Instead, the evidence supports the
empirical premise that financial-state similarity can coexist with
valuation-relevant neural-state variation associated with subsequent valuation
divergence. This is consistent with the state-separation mechanism permitting
neuro-evaluative temporal divergence, without directly identifying subjective
financial time.

\section{Discussion}
\label{sec:discussion}

This paper develops a dynamic neurofinance theory in which subjective
financial time is generated by evaluative progression along
financial--neuro-evaluative histories rather than imposed exogenously. Unlike
the established programme of identifying neural correlates of reward, risk,
loss, framing, and subjective value
\citep{CamererLoewensteinPrelec2005,KnutsonAdamsFongHommer2001,
KuhnenKnutson2005,KableGlimcher2007,PreuschoffBossaertsQuartz2006,
DeMartinoKumaranSeymourDolan2006,TomFoxTrepelPoldrack2007,
MiendlarzewskaKometerPreuschoff2019}, the analysis asks what temporal structure
follows when valuation-related neural state evolves jointly with financial
state on \(x=(x_F,x_E)\).

Theorem~\ref{thm:existence} and Proposition~\ref{prop:temporal-process}
establish that admissible histories generate
\[
\tau_X(t)
=
\tau_0+\int_0^t\omega_{X(s)}(\dot X(s))\,ds,
\]
with
\[
\dot\tau_X(t)=\omega_{X(t)}(\dot X(t)).
\]
Calendar time remains the parameter of physical motion, whereas \(\tau_X\)
represents valuation-induced temporal progression. This distinction is
consistent with evidence that temporal experience varies with attention,
affect, interoception, memory, and distributed neural dynamics
\citep{BuhusiMeck2005,WittmannVanWassenhove2009,Wittmann2013,
TsaoYousefzadehMeckMoserMoser2022,Buzsaki2026}.

Theorem~\ref{thm:dispersion} identifies \(\Lambda(x)\) with dynamically
realisable subjective-time rates and
\(\Delta_\tau(x)=\operatorname{diam}\Lambda(x)\) with their dispersion.
Valuation heterogeneity documented for risk, loss, framing, affect, and
physiological state
\citep{DeMartinoCamererAdolphs2010,DeMartinoKumaranSeymourDolan2006,
PreuschoffQuartzBossaerts2008,CoatesHerbert2008,
SapienzaZingalesMaestripieri2009,PivaEtAl2019Dataset}
therefore has temporal consequences when it corresponds to dynamically
realisable valuation motion.

The central result, Theorem~\ref{thm:neuro-temporal-divergence}, isolates the
neuro-evaluative component. For fixed financial motion,
\(\Delta_{\tau,E}(x\mid u_F)>0\) and matched neuro-evaluative realisability can
generate \(X_F^+(t)=X_F^-(t)\) while
\(\tau_{X^+}(t)\neq\tau_{X^-}(t)\). At the common initial state, the financial
term cancels and the difference in temporal rates is generated by
neuro-evaluative motion. Corollary~\ref{cor:neuro-nonreducibility} consequently
shows that financial equivalence need not imply temporal equivalence.

The empirical analysis examines the state-separation premise underlying this
result rather than subjective financial time itself. \citet{PivaEtAl2019}
identify a medial-prefrontal valuation architecture in which relative
subjective value is positively associated with vmPFC activity and negatively
associated with dmPFC activity, with dmPFC carrying valuation information
across Self/Other and intertemporal/risky-choice contexts. This motivates
\(\nu=z(\mathrm{vmPFC})-z(\mathrm{dmPFC})\) as an empirical measure of relative
medial-prefrontal valuation activity. It is not identified with the theoretical
state \(x_E\), which may be multidimensional.

The matched-state analysis provides proof-of-concept evidence that similarity
in observable financial state need not eliminate valuation-relevant neural
variation. Participant-level coefficients are heterogeneous, with three of
five positive. Fisher combination gives \(p_{\rm Fisher}=0.0216\), which is
interpreted as descriptive evidence against the joint null rather than
evidence of a homogeneous population effect. Participant~03 provides the
clearest within-participant result, with \(\rho=0.253\) and
\(p_{\rm perm}=0.001\). Residual financial-distance adjustment gives
\(\rho=0.251\), while restriction to the closest \(50\%\) of financial matches
retains \(\rho=0.207\). TPJ produces a weaker but positive association for this
participant, with \(\rho_{\rm TPJ}=0.205\). The medial-prefrontal coordinate
therefore produces the larger association, although the positive TPJ result
precludes a claim of anatomical selectivity.

This interpretation complements intertemporal-choice research showing that
objective delay affects valuation circuitry and that discounted subjective
value is represented in medial-prefrontal and striatal systems
\citep{McClureLaibsonLoewensteinCohen2004,KableGlimcher2007,
WittmannPaulus2008}. The cross-context valuation evidence of
\citet{PivaEtAl2019} similarly examines how externally specified temporal and
decision conditions shape valuation. The present framework reverses this
direction. Evolving valuation dynamics can themselves contribute to subjective
temporal structure.

Theorem~\ref{thm:local-structure} distinguishes instantaneous temporal
dispersion from historical path dependence. The quantities \(\Lambda(x)\),
\(\Delta_\tau(x)\), and \(\Delta_{\tau,E}(x\mid u_F)\) characterise
instantaneous rates and their dispersion, while \(d\omega\) governs the local
geometric obstruction to path independence through
\[
\mathcal T_\omega(\gamma_1)-\mathcal T_\omega(\gamma_2)
=
\int_S d\omega.
\]
When the relevant curves are admissible histories, this difference represents
dynamically realised path dependence in subjective financial time.
Corollary~\ref{cor:unique-path-independent-time} recovers locally single-rate,
path-independent time when \(\Delta_\tau=0\) and \(d\omega=0\).
Corollary~\ref{cor:calendar-time} further recovers homogeneous calendar time
when \(\Lambda(x)=\{1\}\). Conventional time therefore appears as a limiting
case of the model.

The framework extends geometric approaches to rational agency
\citep{Stiefenhofer2026GeometryTime,Stiefenhofer2026EconomicTemporality,
Stiefenhofer2026KantianTime,Stiefenhofer2025EthicalConsumption}
by embedding evaluative temporal progression in a neurofinancial state space
and relating temporal multiplicity to dynamically realisable valuation
trajectories. The empirical contribution is more limited. Valuation-relevant
neural variation remains observable after financial-state matching, providing
proof-of-concept evidence for the distinction between financial and
neuro-evaluative state.

The empirical limitations are substantial and define the scope of this
evidence. The sample contains five participants, participant-level effects are
heterogeneous, and the clearest association occurs in Participant~03. The
source experiment measures relative subjective value rather than subjective
temporal experience. Consequently, \(\nu\neq x_E\),
\(\Delta SV\neq\dot\tau\), and matched observable financial states are not
identical realised financial trajectories. The analysis does not estimate
\(\omega\), \(\Delta_\tau\), \(\Delta_{\tau,E}\), or \(d\omega\) and therefore
does not constitute a direct empirical test of the temporal results.

A direct test would combine the valuation architecture of
\citet{PivaEtAl2019} with longitudinal measurement of financial state,
valuation-related neural state, subjective valuation, and temporal experience.
Holding realised financial trajectories fixed would permit direct examination
of whether distinct neuro-evaluative trajectories generate different
subjective-time rates. Repeated histories reaching common financial endpoints
would additionally permit examination of the path-dependent temporal
accumulation characterised by Theorem~\ref{thm:local-structure}.

The theoretical result establishes
\[
X_F^+=X_F^-
\not\Rightarrow
\tau_{X^+}=\tau_{X^-}.
\]
The empirical analysis establishes the narrower result that closely matched
observable financial states need not eliminate valuation-relevant neural
variation associated with subsequent valuation dynamics. The empirical result
does not establish subjective-time divergence, but demonstrates the
neurofinancial state separation required for the theoretical mechanism to have
empirical relevance.

\section{Conclusion}
\label{sec:conclusion}

This paper develops a continuous-time neurofinance theory in which subjective
financial time is generated by evaluative progression along
financial--neuro-evaluative dynamics rather than imposed solely as an external
clock. On the joint state space \(x=(x_F,x_E)\), realised evaluative
progression is represented by the path integral of the valuation one-form
\(\omega\) along admissible histories. The framework reverses conventional
intertemporal analysis by asking not only how time affects valuation, but what
temporal structure follows from evolving valuation dynamics.

The theory identifies three mechanisms. Theorem~\ref{thm:dispersion}
characterises realised temporal-rate dispersion by \(\Delta_\tau\).
Theorem~\ref{thm:neuro-temporal-divergence} establishes that under matched
financial dynamics \(X_F^+(t)=X_F^-(t)\) need not imply
\(\tau_{X^+}(t)=\tau_{X^-}(t)\). Theorem~\ref{thm:local-structure} separates
instantaneous rate heterogeneity from path dependence governed locally by
\(d\omega\). Financial equivalence therefore need not imply temporal
equivalence.

The empirical analysis provides proof-of-concept evidence for the required
neural state separation. From \(1{,}183\) behavioural--fMRI observations,
\(1{,}126\) valuation transitions generate \(798\) financially matched pairs.
For Participant~03, medial-prefrontal valuation-state separation
\(\nu=z(\mathrm{vmPFC})-z(\mathrm{dmPFC})\) predicts subsequent valuation
divergence with \(\rho=0.253\) and \(p_{\rm perm}=0.001\). The coefficient
remains \(0.251\) after residual financial-distance adjustment and \(0.207\)
for the closest \(50\%\) of matches. Effects are heterogeneous across
participants, with \(p_{\rm Fisher}=0.0216\).

These data provide proof-of-concept evidence for neural state separation rather
than subjective financial time itself. The empirical coordinate satisfies \(\nu\neq x_E\),
neither \(\dot{\tau}\) nor \(\Delta_{\tau,E}\) is observed, and matched
financial states are not identical continuous trajectories. Direct
identification requires joint longitudinal measurement of financial motion,
neural valuation state, subjective valuation, and subjective temporal
progression.

The central implication is that valuation-related neural dynamics may
contribute not only to financial choice but also to the temporal structure
within which financial experience unfolds. If confirmed directly, valuation
would not merely occur \emph{in} subjective financial time. Evaluative
progression through the evolving neuro-evaluative state would contribute to
its construction.

\section*{Disclosure of Interests}

The author declares no competing interests.

\appendix
\section{Appendix}
\label{app:appendix}

The analysis uses OpenNeuro \texttt{ds001882}, version 1.0.5
\citep{PivaEtAl2019Dataset}, comprising \(5\) participants,
\(40\) functional runs, and \(1{,}200\) candidate trials.

\begin{table}[htbp]
\centering
\caption{Data construction}
\label{tab:data-lineage}
\small
\begin{tabular}{lrrl}
\toprule
Stage & Available & Retained & Criterion \\
\midrule
Participants  & 5     & 5     & Completed preprocessing \\
Runs          & 40    & 40    & Valid BOLD and motion data \\
Trials        & 1,200 & 1,183 & \(\max FD\leq0.5\) mm; finite ROI responses \\
Transitions   & 1,183 & 1,126 & Consecutive within-run trials \\
Matched pairs & 1,126 & 798   & Nearest-state matching \\
\bottomrule
\end{tabular}
\end{table}

Retained trials and matched-pair counts by participant were
\((236,237,237,239,234)\) and \((154,162,166,163,153)\), respectively. Trial-level PSC was extracted from 5-mm MNI spheres at vmPFC \((0,46,-4)\),
dmPFC \((8,22,48)\), and TPJ \((-58,-40,24)\), with \(TR=2.0\) s, baseline
\([-2,0)\) s, and response window \([4,8)\) s. The neural coordinate was
\(\nu_{irt}=z(\mathrm{vmPFC}_{irt})-z(\mathrm{dmPFC}_{irt})\). Financial matching used standardised amount and delay coordinates within
participant, run, and agency condition. Participant median matching distances
were \((1.185,1.227,1.231,1.309,1.211)\).

\begin{table}[htbp]
\centering
\caption{Computational specification}
\label{tab:reproducibility-specification}
\small
\begin{tabular}{p{3.2cm}p{10.2cm}}
\toprule
Component & Specification \\
\midrule
ROIs
& vmPFC \((0,46,-4)\), dmPFC \((8,22,48)\), TPJ \((-58,-40,24)\);
5-mm MNI spheres \\

PSC
& Baseline \([-2,0)\) s; response \([4,8)\) s \\

Acquisition
& \(TR=2.0\) s; 219 scans \\

Motion QC
& Response-window \(\max FD\leq0.5\) mm \\

Merge
& Participant--run--trial; duplicate keys rejected; no neural imputation \\

Financial match
& Standardised amount/delay coordinates; participant--run--agency conditioning;
nearest Euclidean neighbour \\

Neural index
& \(z(\mathrm{vmPFC})-z(\mathrm{dmPFC})\) \\

Primary statistic
& Participant-specific Spearman correlation \\

Inference
& \(5{,}000\) within-cell permutations \\

Robustness
& Rank-partial adjustment; closest \(50\%\) of matches; TPJ comparator \\
\bottomrule
\end{tabular}
\end{table}

For participant \(s\), permutation inference used

\begin{equation}
p_s
=
\frac{
1+\sum_{b=1}^{5000}
\mathbf{1}
\left\{
\rho_s^{(b)}
\geq
\rho_s^{\mathrm{obs}}
\right\}
}{
5001
}.
\label{eq:appendix-permutation-p}
\end{equation}

The computational lineage was therefore
\[
1{,}200
\rightarrow
1{,}183
\rightarrow
1{,}126
\rightarrow
798.
\]

\begin{table}[htbp]
\centering
\caption{Participant-level inference and robustness}
\label{tab:appendix-participant-results}
\small
\begin{tabular}{lrrrrrr}
\toprule
Participant &
Pairs &
\(\rho\) &
\(\rho_{\mathrm{partial}}\) &
\(p_{\mathrm{perm}}\) &
\(\rho_{\mathrm{strict}}\) &
\(\rho_{\mathrm{TPJ}}\) \\
\midrule
01 & 154 &  0.006 &  0.012 & 0.448 & -0.006 &  0.009 \\
02 & 162 & -0.087 & -0.089 & 0.624 &  0.043 &  0.092 \\
03 & 166 &  \textbf{0.253} & \textbf{0.251} &
\textbf{0.001} & \textbf{0.207} & 0.205 \\
04 & 163 & -0.053 & -0.053 & 0.734 &  0.002 & -0.010 \\
05 & 153 &  0.097 &  0.106 & 0.139 & -0.114 & -0.096 \\
\bottomrule
\end{tabular}
\end{table}

The five one-sided participant-level permutation probabilities give
\(p_{\mathrm{Fisher}}=0.0216\). For Participant~03, restricting matching to
the closest \(50\%\) reduces the sample from \(166\) to \(83\) pairs, with
\(\rho_{\mathrm{strict}}=0.207\) and
\(\rho_{\mathrm{TPJ,strict}}=0.143\).


\end{document}